\documentclass[journal]{IEEEtran}

\usepackage[utf8]{inputenc} 
\usepackage[T1]{fontenc}
\usepackage{url}
\usepackage{ifthen}
\usepackage{cite}
\usepackage[cmex10]{amsmath} 
\usepackage{amssymb}
\usepackage{graphicx}
\usepackage{epstopdf}
\usepackage{epsfig}
\usepackage{multirow}
\usepackage{amsthm}

\usepackage{algorithm}
\usepackage{algpseudocode}
\usepackage{csquotes}
\usepackage{caption}

\newtheorem{thm}{Theorem}
\newtheorem{lem}{Lemma}

\newtheorem{cor}{Corollary}

\newtheorem{rem}{Remark}

\begin{document}
\title{Decentralized and Online Coded Caching with Shared Caches: Fundamental Limits with Uncoded Prefetching} 

\author{Elizabath Peter, and
        B. Sundar Rajan,~\IEEEmembership{Fellow,~IEEE}
\thanks{E. Peter and B. S. Rajan are with the Department
of Electrical Communication Engineering, Indian Institute of Science, Bangalore, 560012,
India (e-mail: elizabathp@iisc.ac.in, bsrajan@iisc.ac.in).}% <-this % stops a space
\thanks{This work was supported partly by the Science and Engineering Research Board (SERB) of Department of Science and Technology (DST), Government of India, through J.C. Bose National Fellowship to B. Sundar Rajan.}
}

\maketitle 
\begin{abstract}
Decentralized coded caching scheme, introduced by Maddah-Ali and Niesen, assumes that the caches are filled with no coordination. This work identifies a decentralized coded caching scheme -under the assumption of uncoded placement- for shared cache network, where each cache serves multiple users. Each user has access to only a single cache and the number of caches is less than or equal to the number of users. For this setting, we derive the optimal worst-case delivery time for any user-to-cache association profile where each such profile describes the number of users served by each cache. The optimality is shown using an index-coding based converse. Further, we improve the delivery scheme to accommodate redundant demands. Also, an optimal linear error correcting delivery scheme is proposed for the worst-case demand scenario. Next, we consider the Least Recently Sent (LRS) online coded caching scheme where the caches need to be updated based on the sequence of demands made by the users. Cache update happens if any of the demanded file was not partially cached at the users. The update is done by replacing the least recently sent file with the new file. But, the least recently sent file need not be unique. In that case, there needs to be some ordering of the files which are getting partially cached, or else centralized coordination would have to be assumed which does not exist. If each user removes any of the least recently used files at random, then the next delivery phase will not serve the purpose. A modification is suggested for the scheme by incorporating the ordering of files. The fix is suggested so that the results on the error-correction for the online coded caching scheme introduced in [N. S. Karat, K. L. V. Bhargav and B. S. Rajan, \enquote{On the Optimality of Two Decentralized Coded Caching Schemes With and Without Error Correction,} \textit{IEEE ISIT}, Jun. 2020] are not affected.  Moreover, all the above results with shared caches are extended to the online setting.

\end{abstract}
%\begin{IEEEkeywords}

%\end{IEEEkeywords}
%%%%%%%%%%%%%%%%%

\section{Introduction}
 \label{section:intro}
 %\subfile{sections/introduction}
 
 Caching is a technique to combat the traffic congestion experienced during peak hours, by prefetching parts of the contents into memories distributed across the network during off-peak times. Caching is done in two phases: \textit{placement (or prefetching) phase} and \textit{delivery phase} \cite{r1}. In the placement phase, portions of the contents are stored in the cache memories in uncoded or coded form, without knowing the future requests. In the delivery phase, users reveal their demands to the server and the server has to find effective transmission scheme by utilizing the cached contents. 

The seminal work in \cite{r1} showed that coding in the delivery phase achieves substantial gains over uncoded transmissions by considering a network where a server has access to a database of $N$ equally-sized files and is connected to $K$ users via a bottleneck link. Each user has its own cache of size equal to $M$ files. The scheme in \cite{r1} is a centralized coded caching scheme, which means the placement phase is centrally coordinated. For the same network model, a decentralized coded caching scheme was introduced in \cite{r3} which assumes no coordination in the placement phase, where contents are placed independently and randomly across different caches.

The coded caching problem has been studied in a variety of settings. In \cite{r2}, a cache-aided heterogeneous network model is considered where a server communicates to a set of users with the help of smaller nodes, in which these helper nodes serve as caches that is shared among many users. Each user has access to a single helper cache, which is potentially accessed by multiple users. For this setup, a centralized coded caching scheme based on uncoded placement was proposed in \cite{r2}. In this work, we identify a decentralized coded caching scheme for the shared cache network, under the constraint of uncoded prefetching. We further extend it to the online case where the cache contents are updated on the fly based on the sequence of requests made by the users \cite{online}. 

Error correcting coded caching schemes were introduced for different settings in \cite{r6}-\cite{nj}. In this work, we also consider the case of erroneous shared link during the delivery phase and give an optimal error correcting delivery scheme for the shared cache system in both offline and online scenarios.

\subsection{Related Work}
The decentralized coded caching problem for shared cache was first addressed in [2, Section V.B], which considered a setting where each cache serves equal number of users and each user requests a different file. The caching scheme described in this work accounts any kind of user-to-cache association and any demand type.

The centralized caching scheme in \cite{r2} provides the optimal delivery time for distinct demands, under uncoded prefetching. To accommodate redundant user demands, a new delivery scheme was introduced in \cite{r7} which provides an improved performance for non-distinct demands. A caching scheme based on coded placement was proposed in \cite{r8}, which outperforms the scheme in \cite{r2} in certain memory regimes.

In \cite{r7}, the authors also consider the case when the delivery phase is error-prone and provide an optimal linear error correcting delivery scheme. For shared cache with small memory sizes, \cite{an} provides an optimal linear error correcting scheme. The results of \cite{r2} have been further extended in \cite{r4}, in which the fundamental limits are derived for the case when server is equipped with multiple antennas. 

\subsection{Contributions}
In the context of coded caching with shared caches, our contributions are as follows:
\begin{itemize}
    \item A decentralized coded caching scheme based on uncoded prefetching is proposed, and is shown to be optimal for distinct demand case using index coding techniques. A closed form expression for the delivery time is derived for any user-to-cache association and any demand vector. (Sections \ref{section:decen} and \ref{section:converse}).
    
    \item The decentralized coded caching algorithm for the dedicated cache network \cite{r3} is obtained as a special case of the scheme in this paper.
    
    \item The delivery scheme introduced for uniform user-to-cache association in [2, Section V.B] is shown to be optimal. (Appendix).
    
    \item For the online coded caching scheme in \cite{online}, an ordering is introduced for the initial placement of files to alleviate the ambiguity arising when the least recently sent file is not unique. (Section \ref{section:nujoom}).
    
    \item Generalizes the online coded caching algorithm in \cite{online} to shared caches and derives an expression for the optimal delivery time in distinct demand case. (Section \ref{section:onl}).
    
    \item An optimal linear error correcting delivery scheme is found for the shared caching problem with decentralized placement and the corresponding expression for the worst-case delivery time is obtained. (Section \ref{section:errorc}).
    
  \item For online coded caching with shared caches, an optimal linear error correcting delivery scheme is identified for distinct requests, and an expression for the delivery time in the presence of finite number of errors is obtained. (Section \ref{section:errorc}).

 \end{itemize}
 
 \subsection{Notation}
For any integer $n$, $[n]$ denotes the set $\{1,2,...,n\}$, For a set $\mathcal{S} \subseteq [\Lambda]$, $W^{n}_\mathcal{S}$ represents the bits of file $W^n$ cached exclusively at the helper nodes present in set $\mathcal{S}$. $|\mathcal{S}|$ denotes the cardinality of set $\mathcal{S}$. Binomial coefficients are denoted by $\binom{n}{k}$, where $\binom{n}{k} \triangleq \frac{n!}{k!(n-k)!}$ and $\binom{n}{k}$ is zero for $n < k$. $\mathbb{F}_q$ represents a finite field with $q$ elements. $\mathbb{e}_i \in \mathbb{F}_q^n$ represents a unit vector having a one at $i^{th}$ position and zeros elsewhere and $N_q[k,d]$ denotes the length of an optimal linear error-correcting code over $\mathbb{F}_q$ with dimension $k$ and minimum distance $d$. 

\subsection{Paper Outline}
In Section \ref{section:model}, we provide a description of the system model considered, followed by some preliminaries about index coding and online caching in Section \ref{section:prelim}. Section \ref{section:mainres} presents the
main results and Section \ref{section:decen} describes the caching scheme. The optimality of the scheme is shown in Section \ref{section:converse}. Section \ref{section:nujoom} considers the online setting and presents the modification suggested. Section \ref{section:onl} describes the online coded caching scheme for shared caches and Section \ref{section:errorc} gives error-correcting delivery schemes for both scenarios. Section \ref{section:ex} presents some examples.

\section{System Model}
 \label{section:model}
 %\subfile{sections/systemmodel}
 We consider a network where a server has a library of $N$ files $\{W^1, W^2,\ldots, W^n\}$, each of size $F$ bits, connected to $K$ requesting users through an error-free shared broadcast link. Each user has access to one of the $ \Lambda \leq K$ helper nodes/caches, each of size $M$ units of file and each cache serves an arbitrary number of users. 

The communication process involves the following three phases:\\
\textit{a) Placement phase:} In the placement phase, the helper nodes have access to the library of $N$ files and store contents from it in a decentralized manner, which is explained in detail in Section \ref{section:decen}.A.\\
\textit{b) User-to-cache association phase:} The placement phase is followed by an additional phase where each of the $K$ requesting  users is given access to exactly one helper cache. That is, each cache $\lambda \in [\Lambda]$ is connected to a set of users $\mathcal{U_\lambda} \subseteq [K]$ such that these sets form a partition of the set of $K$ users. The user-to-cache association is represented as:\\
\begin{equation*}
     \mathcal{U} = \{\mathcal{U}_1, \mathcal{U}_2,\ldots, \mathcal{ U}_\Lambda\}.
\end{equation*}
\noindent This assignment is oblivious of the cached contents and the demands of the users. For a given $\mathcal{U}$, define the association profile (or user profile) $\mathcal{L}$ as
\begin{equation*}
    \mathcal{L} \triangleq (\mathcal{L}_1,\ldots,\mathcal{L}_i,\ldots, \mathcal{L}_\Lambda)
\end{equation*}
\noindent where, $\mathcal{L}_i$ is the number of users associated to cache $i$. Naturally, $\sum_{i=1}^{\Lambda}\mathcal{L}_i = K$. Without loss of generality, we assume $\mathcal{L}_i\geq\mathcal{L}_j$ for $i \leq j$ and $\mathcal{U}_i$ to be an ordered set.
\\

\begin{figure}[t!]
 \begin{center}
 \captionsetup{justification=centering}
 \includegraphics[width=\columnwidth]{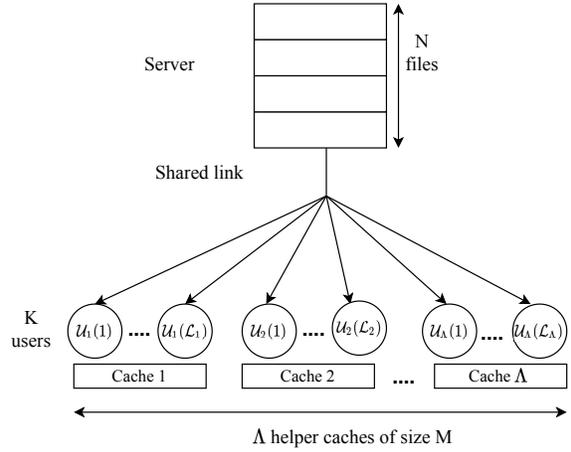}
 \caption{Caching setting considered in this work.}
 \end{center}
\end{figure}

\noindent \textit{c) Delivery phase:} The delivery phase starts when each user $k \in [K]$ demands a file $W^{d_k}$. Upon receiving the demand vector $d = (d_1, d_2,..., d_K)$, the server broadcasts a message over the shared link to the users. Each user should be able to obtain its requested file using the cached contents and the received message.\\ 

\noindent \textit{Performance measure}: Our measure of interest is the delivery time required to satisfy a demand vector $d$, for a given user-to-cache association $\mathcal{U}$.  Different user-to-cache associations $\mathcal{U}$ can yield the same profile $\mathcal{L}$. Therefore, we express the delivery time in terms of $\mathcal{L}$ and let $T_D(\mathcal{L},d)$ denote the delivery time needed to satisfy any demand $d$. The worst-case delivery time for a given $\mathcal{L}$ corresponds to the maximum value of $T_D(\mathcal{L},d)$ over all the demand types. Our aim is to characterize the minimum possible worst-case delivery time $T_D^*(\mathcal{L})$, under the assumption of decentralized uncoded placement. For non-distinct demands, an improved delivery time can be obtained by following the transmission scheme described in Section \ref{section:decen}.E.

\section{Preliminaries}
 \label{section:prelim}
% \subfile{sections/preliminaries}
In this section, we briefly review some of the concepts from error correction for index coding \cite{r5}, which are relevant in proving the optimality of our caching scheme. We also discuss the online scheme in \cite{online} in brief.

\subsection{Index Coding Problem}
The index coding problem with side-information was first considered in \cite{iscod}. The setup is as follows: there is a server with $n$ messages $\{x_1,\ldots, x_n\}$, where $x_i \in \mathbb{F}_q$, $\forall i \in [n]$. There are $K$ receivers and each of them possesses a subset of messages $\{x_j\}_{j \in \mathcal{X}_i}$ as side-information where $\mathcal{X}_i$ represents set of indices of the messages known to $i^{th}$ receiver. Each receiver wants $x_{f(i)}$ where $f: [K] \rightarrow[n]$ is a function that maps receivers to the index of the message demanded by the them and $f(i) \notin \mathcal{X}_i$ \cite{r5}. (A receiver demanding multiple messages can be equivalently seen as that many receivers each demanding a single message and has the same side-information). With the knowledge of the side-information available to each receiver, the server has to satisfy the demands of all the receivers in minimum number of transmissions.

An instance of the index coding problem can be described by a directed hypergraph $\mathcal{H}=(\mathcal{V}, \mathcal{E}_{\mathcal{H}})$ where $\mathcal{V}=[n]$ is the set of vertices and $\mathcal{E}_{\mathcal{H}}$ is the set of hyperedges \cite{index}. Vertices represent the messages and each hyperedge represents a receiver. The min-rank of the hypergraph $\mathcal{H}$  corresponding to an index coding problem is defined as \cite{r5}
\begin{equation*}
\kappa({\mathcal{I}}) = min\{rank_q(\{v_i + e_{f(i)}\}_{i \in [K]}): v_i \in \mathbb{F}_q^n, v_i \triangleleft \mathcal{X}_i\}.
\end{equation*}
The min-rank corresponds to the smallest length of the scalar linear index code for the problem. 

For an undirected graph, $\mathcal{G} = (\mathcal{V}, \mathcal{E})$, a subset of vertices $\mathcal{S} \subseteq \mathcal{V}$ is called an independent set if $\forall u, v \in \mathcal{S}$, $(u, v) \notin \mathcal{E}$. The cardinality of largest independent set in $\mathcal{G}$ is called the independence number of $\mathcal{G}$.
In \cite{r5}, the notion of independence number is extended to the directed hypergraph corresponding to an index coding problem. For each receiver $i \in [K]$, define the sets 
\begin{equation}
\mathcal{Y}_i \triangleq [n] \backslash \big(\{f(i)\}\cup \mathcal{X}_i\big)
\end{equation}
and
\begin{equation}
    \mathcal{J}(\mathcal{I}) \triangleq \cup_{i\in[K]}\{\{f(i)\} \cup Y_i : Y_i \subseteq\mathcal{Y}_i\}.
    \label{ind}
\end{equation}

A subset $H$ of $[n]$ is called a generalized independent set in $\mathcal{H}$ if every non-empty subset of $H$ belongs to $\mathcal{J}(\mathcal{I})$. The size of the largest independent set in $\mathcal{H}$ is called the generalized independence number, and denoted by $\alpha(\mathcal{H})$. 

For any index coding problem, the following relation exists \cite{r6}
\begin{equation}
\alpha(\mathcal{H}) \leq \kappa(\mathcal{H}).
\label{bound}
\end{equation}
For a directed graph, generalized independence number is same as the size of the maximum acyclic induced subgraph of $\mathcal{H}$.

\subsection{Error Correcting Index Code(ECIC)}
An index code is said to correct $\delta$-errors if each receiver is able to decode its message correctly even if at most $\delta$ transmissions went erroneous. A $\delta$-error correcting index code is represented as $(\delta, \mathcal{H})$-ECIC. An optimal linear $(\delta, \mathcal{H})$-ECIC over $\mathbb{F}_q$ is a linear $(\delta, \mathcal{H})$-ECIC over $\mathbb{F}_q$ with minimum possible length $\mathcal{N}_q[\mathcal{H}, \delta]$. According to \cite{r5}, an optimal linear $(\delta, \mathcal{H})$-ECIC over $\mathbb{F}_q$ must satisfy
\begin{equation}
    N_q[\alpha(\mathcal{H}), 2\delta+1] \leq  \mathcal{N}_q[\mathcal{H}, \delta] \leq   N_q[\kappa(\mathcal{H}), 2\delta+1].
    \label{ebound}
\end{equation}
The lower bound is known as the $\alpha$-bound and the upper bound is known as the $\kappa$-bound. The upper-bound is obtained by concatenating an optimal linear index code with an optimal linear error correcting code. When $\alpha(\mathcal{H}) = \kappa(\mathcal{H})$, both the bounds meet and concatenation gives optimal ECICs. 

\subsection{Corresponding Index Coding Problem for a Coded Caching Scheme}
For a fixed placement and for a fixed demand, the delivery phase of a coded caching problem is an index coding problem. In the case of shared caches, the caching problem is defined when the user-to-cache association and demands are revealed. The side-information available to each receiver depends on the user-to-cache association. The index coding problem corresponding to the shared caching scheme with the decentralized placement  $\mathcal{M}_D$, user-to-cache association $\mathcal{L}$ and demand $d$ is denoted by $\mathcal{I}(\mathcal{M}_D,\mathcal{L},d)$. 
For the index-coding problem $\mathcal{I}(\mathcal{M}_D,\mathcal{L},d)$, the min-rank and the  generalized independence number are denoted as $\kappa(\mathcal{M}_D ,\mathcal{L},d)$ and $\alpha(\mathcal{M}_D, \mathcal{L},d)$, respectively.

\subsection{Online Coded Caching Scheme: LRS (Least Recently Sent) Scheme}

The online coded caching scheme- known as the coded LRS caching algorithm-was introduced in \cite{online}, for a network with each user having its own cache. At a time-slot $t$, the coded LRS algorithm consists of two phases: a delivery phase and a cache update phase. The demand of the user $k$ at time $t$, is denoted by $W^{d_t(k)}$. The users issue a sequence of requests from a set of popular files $\mathcal{N}_t$, uniformly at random ( $|\mathcal{N}_t|=N$). The set of popular files change over the course of time. A basic Markov model is assumed for renewing the set of popular files. The set $\mathcal{N}_t$ evolves from the set $\mathcal{N}_{t-1}$ using an arrival/departure process. Consequently, at the time of demand,
some of the requested files may not be cached in the memories. Those files need to be sent uncoded over the link and the remaining requests are served using the decentralized scheme in \cite{r3}.

After the delivery phase, the caches are updated as follows. In each time slot $t$, each user maintains a list of $N^{'}= \beta N$ partially cached files for some $\beta \geq 1$. At time $t$, after the delivery, if a requested file $W^{d_t(k')}$ is not currently partially cached, each user replaces the least-recently sent file with $MF/N^{'}$ randomly chosen bits from the file $W^{d_t(k')}$. This is possible, since the file  $W^{d_t(k')}$ was sent uncoded during the delivery. The eviction policy guarantees that the number of cached files remains the same at all instants.

\section{Main Result}
 \label{section:mainres}
 %\subfile{sections/mainresult}
 \noindent The main result is presented in the following theorem.
\begin{thm}
  For a coded caching problem with $N$ files, each of size $F$ bits ($F$ is large enough) and $K$ users sharing $\Lambda \leq K$ helper caches of size $MF$ bits, where $M \in [0,N]$, the  optimal worst-case delivery time for a given association profile $\mathcal{L}$, under decentralized uncoded placement is 
  \label{thm1}
  \begin{equation}
  \begin{aligned}[f]
     T_D^{*}(\mathcal{L}) & = 
     \left(\frac{N-M}{M}\right)\times\\
     & \displaystyle\sum_{s=1}^{\Lambda}\left[
\displaystyle\sum_{n=1}^{\Lambda-(s-1)} \mathcal {L}_n\binom{\Lambda-n}{s-1}\right]
\left(\frac{M}{N}\right)^s\left(1-\frac{M}{N}\right)^{\Lambda-s}. \label{eq1}
\end{aligned}
 \end{equation}
 
  \end{thm}

\begin{proof}\renewcommand{\qedsymbol}{}
   The achievability and converse of \eqref{eq1} are presented in Section IV and Section V, respectively.
\end{proof} 

\begin{cor}
If the user-to-cache association is uniform, then $\mathcal{L}_{unif.}= (\frac{K}{\Lambda},..,\frac{K}{\Lambda})$. The optimal worst-case delivery time becomes
\begin{equation}
  \begin{aligned}[f]
      T_D^{*}(\mathcal{L}_{unif.})  =
     \left(\frac{N-M}{M}\right)\left( \frac{K}{\Lambda}\right)\left[1-\left(1-\frac{M}{N}\right)^{\Lambda}\right].
     \label{eq2}
\end{aligned}
 \end{equation}
 
 \end{cor}
\begin{proof}\renewcommand{\qedsymbol}{}
The proof of transition from \eqref{eq1} to \eqref{eq2} is given in Appendix.
\end{proof} 

\noindent Further in the case of uniform profile, if $\Lambda = K$, \eqref{eq2} reduces to
 \begin{equation}
  \begin{aligned}[f]
      T_D^{*}(\mathcal{L}_{unif.}, \Lambda = K) & =
     \left(\frac{N-M}{M}\right)\left[1-\left(1-\frac{M}{N}\right)^{K}\right].
    \label{eq3}
\end{aligned}
\end{equation}

The expression in \eqref{eq3} is same as the time obtained for a dedicated cache network in decentralized case \cite{r3}.\\

When all the users are assigned to the same helper cache, the delivery time is
\begin{equation}
   \begin{aligned}[f]
      T_D^*(\mathcal{L}) = K\left(1 - \frac{M}{N}\right).
      \label{eq4}
    \end{aligned}
\end{equation}
Derivation of \eqref{eq4} is given in the Appendix. In this case, we obtain only local caching. The centralized scheme in \cite{r2} also exhibits the same performance.

Next, consider online coded caching with shared caches. The users request from a time-varying set of $N$ popular files, and the caches prefetch contents from a set of $N^{'}\geq N$ files in a decentralized fashion. Assume distinct demand case, the optimal delivery time $T^{*}_{online}(\mathcal{L})$ is given in the following theorem.

\begin{thm}
  In an online coded caching problem with $\Lambda$ shared caches, when $U$ number of requested files are not partially cached in the memories, the optimal delivery time required for any profile $\mathcal{L}$ is
  \label{thm2}
  \begin{equation}
  \begin{aligned}[f]
     T_{online}^{*}(\mathcal{L}) & = U +
     \left(\frac{N^{'}-M}{M}\right).\\
     & \displaystyle\sum_{s=1}^{\Lambda}\left[
\displaystyle\sum_{n=1}^{\Lambda-(s-1)} \mathcal {L}_n^{'}\binom{\Lambda-n}{s-1}\right]
\left(\frac{M}{N^{'}}\right)^s\left(1-\frac{M}{N^{'}}\right)^{\Lambda-s}
\label{eqonl}
\end{aligned}
 \end{equation}
\noindent where, $\mathcal{L}^{'}$ is the modified profile obtained after removing the users requesting those $U$ files from its corresponding caches.
  \end{thm}

\begin{proof}\renewcommand{\qedsymbol}{}
   The achievability and converse of \eqref{eqonl} are given in Section \ref{section:onl}.
\end{proof} 

\begin{rem}
  When all the requested files are partially cached, the situation becomes similar to the decentralized scheme with shared caches. The optimal delivery time then reduces to
  \begin{equation}
  \begin{aligned}[f]
     T_{online}^{*}(\mathcal{L}) & =     \left(\frac{N^{'}-M}{M}\right)\times\\
     & \displaystyle\sum_{s=1}^{\Lambda}\left[
\displaystyle\sum_{n=1}^{\Lambda-(s-1)} \mathcal {L}_n\binom{\Lambda-n}{s-1}\right]
\left(\frac{M}{N^{'}}\right)^s\left(1-\frac{M}{N^{'}}\right)^{\Lambda-s}.
\label{eqonl2}
\end{aligned}
 \end{equation}
  
\end{rem}
 
\section{Decentralized Coded Caching Scheme}
 \label{section:decen}
 %\subfile{sections/scheme}
 The formal description of the caching scheme which achieves the performance in Theorem \ref{thm1} is given below.
\subsection{Placement Phase}
Each helper node independently caches $\frac{MF}{N}$ bits of each file, chosen uniformly at random. The placement procedure effectively partitions each file $W^{n}, n \in [N]$ into $2^\Lambda$ subfiles in the following manner:
\begin{equation*}
    W^n = (W^n_\phi, W^n_{\{1\}},..,W^n_{\{\Lambda\}},W^n_{\{1,2\}},...,W^n_{\{1,2,..,\Lambda\}}).
\end{equation*}
 By law of large numbers, for $\mathcal{S} \subseteq [\Lambda]$, $|W^n_\mathcal{S}|$ is given as:
 \begin{equation*}
   |W^n_\mathcal{S}| \approx \left(\frac{M}{N}\right)^{|\mathcal{S}|}\left(1-\frac{M}{N}\right)^{\Lambda-|\mathcal{S}|}F
 \end{equation*}
with probability approaching one for large enough file size $F$. The helper node, $\lambda \in [\Lambda]$ has subfiles of the form $W^n_\mathcal{S}$, where $\lambda \in \mathcal{S}$ for $\forall n \in [N]$.
\subsection{Delivery Phase}
Consider a user-to-cache association $\mathcal{U}$ with a profile $\mathcal{L}$.
The demands of the users are communicated to the server, assume each user requests a different file. For an association profile $\mathcal{L}$, delivery consists of $\mathcal{L}_1$ rounds (similar to the transmission method in \cite{r2}), where each round $j \in [\mathcal{L}_1]$ serves the set of users
\begin{equation*}
    \mathcal{R}_j = \bigcup_{\lambda \in [\Lambda]}(\mathcal{U}_{\lambda}(j) : \mathcal{L}_{\lambda} \geq j)
\end{equation*}
\noindent where, $\mathcal{U}_{\lambda}(j)$ is the $j$-th user in the set $\mathcal{U}_{\lambda}$.

The following algorithm describes the delivery procedure for the demand vector $(d_1,..,d_K)$.

\begin{algorithm}
\renewcommand{\thealgorithm}{1}
\caption{Delivery phase}
\begin{algorithmic}[1]
%\Procedure{Roy}{$a,b$}       \Comment{This is a test}
    \For{\texttt{$j = 1:\mathcal{L}_1$}}
     \For{\texttt{$s = \Lambda, \Lambda-1,...,1$}}
       \For{\texttt{$\mathcal{S} \subseteq [\Lambda]: |\mathcal{S}|=s$}}
         \State  Transmit $\underset{\lambda \in \mathcal{S}: \mathcal{L}_\lambda \geq j}{\bigoplus} W_{\mathcal{S}\backslash\{\lambda\}}^{d_{\mathcal{U}_{\lambda}(j)}}$
      \EndFor
     \EndFor
    \EndFor

\end{algorithmic}
\end{algorithm}
The delivery procedure assumes that the server has the knowledge of the contents placed at each helper node. This is achieved by sending the seed value of the random number generator to the server, which is used for caching the bits of each file. 
In the case of uniform user-to-cache association with $\Lambda=K$, the above algorithm recovers the delivery scheme in \cite{r3}. Hence, it can be viewed as a generalization of the algorithm in \cite{r3}.

\subsection{Decoding procedure}
Consider a transmission corresponding to a particular set   $\mathcal{S}$ in line $3$. For a user $\mathcal{U}_{\lambda}(j)$, the received signal is of the form:
\begin{equation*}
    y_{\mathcal{U}_{\lambda}(j)} = W_{\mathcal{S}\backslash\{\lambda\}}^{d_{\mathcal{U}_{\lambda}(j)}}\underbrace{\underset{\lambda' \in \mathcal{S}\backslash\{\lambda\}: \mathcal{L}_{\lambda'} \geq j} {\bigoplus} W_{\mathcal{S}\backslash\{\lambda'\}}^{d_{\mathcal{U}_{\lambda'}(j)}}}_\text{available as side-information}.
\end{equation*}
Thus, each user $\mathcal{U}_{\lambda}(j)$ can obtain its desired subfiles from the received messages by canceling out the side-information.

\subsection{Delivery time Calculation}
The probability of a bit of a file to be present in one of the helper node/cache is:
\begin{equation*}
    q \triangleq \frac{M}{N}.
\end{equation*}
Then, the expected size of $W_{\mathcal{S}\backslash\{\lambda\}}^{d_{\mathcal{U}_{\lambda}(j)}}$ with $|\mathcal{S}|=s$, where $s \in [\Lambda]$ is:
\begin{equation*}
    F q^{s-1}(1-q)^{\Lambda-(s-1)}.
\end{equation*}

\noindent The number of transmissions associated to each $s$ sized subset in each round $j \in [\mathcal{L}_1]$ is:
\begin{equation*}
    \binom{\Lambda}{s} - \binom{\Lambda - |\mathcal{R}_j|}{s}.
\end{equation*}
Then, the time required for the transmissions corresponding to a particular value of $s$ over all rounds is
\begin{equation*}
   \left(\displaystyle\sum_{j=1}^{\mathcal{L}_1}\binom{\Lambda}{s}-\binom{\Lambda - |\mathcal{R}_j|}{s}\right)q^{s-1}(1-q)^{\Lambda-(s-1)}F.
\end{equation*}
Summing over all values of $s$ gives the normalized delivery time as:
\begin{equation}
  \begin{aligned}[f]
      T_D(\mathcal{L}) & =\\
     &  \displaystyle\sum_{s=1}^{\Lambda} \left(\displaystyle\sum_{j=1}^{\mathcal{L}_1}\binom{\Lambda}{s}-\binom{\Lambda - |\mathcal{R}_j|}{s}\right)q^{s-1}(1-q)^{\Lambda-s+1}\\
      = & \left(\frac{1-q}{q}\right).\\
      &\displaystyle\sum_{s=1}^{\Lambda} \left(\displaystyle\sum_{j=1}^{\mathcal{L}_1}\binom{\Lambda}{s}-\binom{\Lambda - |\mathcal{R}_j|}{s}\right)q^{s}(1-q)^{\Lambda-s}.
      \label{e}
\end{aligned}
 \end{equation}

The term
$ \displaystyle\sum_{j=1}^{\mathcal{L}_1}\binom{\Lambda}{s}-\binom{\Lambda - |\mathcal{R}_j|}{s}$ can be replaced by $ \displaystyle\sum_{n=1}^{\Lambda-(s-1)}\mathcal{L}_n\binom{\Lambda-n}{s-1}$ (see \cite{r4}, Section V). Also, substituting for $q = \frac{M}{N}$ in \eqref{e}, the normalized delivery time becomes

\begin{equation}
  \begin{aligned}[f]
      T_D(\mathcal{L}) &= 
      \left(\frac{N-M}{M}\right)\times\\
      & \displaystyle\sum_{s=1}^{\Lambda} \left(\displaystyle\sum_{n=1}^{\Lambda-(s-1)}\mathcal{L}_n\binom{\Lambda-n}{s-1}\right)\left(\frac{M}{N}\right)^{s}\left(1-\frac{M}{N}\right)^{\Lambda-s}.
     \label{e2}
\end{aligned}
 \end{equation}

This completes the achievability part of the caching scheme. 

\subsection{For non-distinct demands}
In the case of redundant user demands, we consider a slightly improved transmission scheme based on the scheme in [7, Section V]. For a demand vector $d=(d_1, d_2,..., d_K)$, $N_e(d)$ represents the number of distinct files in $d$. Initially, we will address the issue of repeated demands arising from the users connected to the same cache. For a user-to-cache association $\mathcal{U}$ with a profile $\mathcal{L}$, let $N_e(\mathcal{U}_\lambda)$ denote the number of distinct files requested by the users associated to cache $\lambda \in [\Lambda]$. We will consider only $N_e(\mathcal{U}_\lambda)$ users from each cache, since the remaining users in $\mathcal{U}_\lambda$ can obtain their files from the transmissions meant for $N_e(\mathcal{U}_\lambda)$ users. Hence, the users with redundant demands from each cache are not considered in the delivery process resulting in a different user-to-cache association $\mathcal{U}'$. Let the corresponding association profile be $\mathcal{L'}$, and without loss of generality, we assume $\mathcal{L}_1^{'}\geq...\geq\mathcal{L}_\Lambda^{'}$.

The server initially transmits $W^n_{\phi}$ subfiles of the distinct files requested by the users, in uncoded form. The delivery of the remaining subfiles happens in $\mathcal{L}_1^{'}$ rounds. For $j \in [\mathcal{L}_1^{'}]$, define a set $\mathcal{R}_j^{'}$ as follows
\begin{equation*}
    \mathcal{R}_j^{'} = \bigcup_{\lambda \in [\Lambda]}(\mathcal{U}_{\lambda}^{'}(j) : \mathcal{L}_{\lambda}^{'}\geq j).
\end{equation*}

Redundancy in demands can still exist amongst the users belonging to different caches. Let $N_e(j)$ be the number of distinct requests in round $j$. In each round, the server arbitrarily selects a subset of $N_e(j)$ users from $\mathcal{R}_j^{'}$, requesting different files. These set of users are referred to as leaders and is denoted by $\mathcal{Q}_j$. Algorithm 1 works here also, with minor changes in line $3$ and $4$. In line $3$: instead of $s\in [\Lambda]$, we take $s \in [\Lambda]\backslash \{1\}$ due to the initial transmission of $W^n_\phi$ subfiles. And in line $4$: the server transmits the message only if it benefits at least one user in $\mathcal{Q}_j$, otherwise there is no transmission. 

\begin{thm}
When all the demands are not distinct, the delivery time achieved is
\begin{equation}, 
  \begin{aligned}[f]
     T_D(\mathcal{L}) = N_e(d)(1-q)^{\Lambda} + &\displaystyle\sum_{s=2}^{\Lambda} \left[\displaystyle\sum_{j=1}^{\mathcal{L}_1^{'}}\binom{\Lambda}{s}-\binom{\Lambda - N_e(j)}{s}\right].\\
     & q^{s-1}(1-q)^{\Lambda-s+1}
     \label{red}
  \end{aligned}
 \end{equation}
\noindent where, $q \triangleq M/N$.
 \end{thm}
 
 \begin{proof}\renewcommand{\qedsymbol}{}
 
 The first term in the right-hand side of \eqref{red} corresponds to the transmission of $W^n_\phi$ subfiles of the $N_e(d)$ distinct files. The server greedily broadcasts the remaining subfiles only if it directly helps one leader user. That is, in each round $j\in[\mathcal{L}_1^{'}]$ ($\mathcal{L}^{'}$ is the user profile obtained after removing the redundancies within a cache), the server makes the following number of transmissions $\binom{\Lambda}{s}-\binom{\Lambda-N_e(j)}{s}$, $\forall s \in \{2,\ldots,\Lambda\}$. This results in the second term of the right hand side of \eqref{red}.
 
 The set of leader users can directly decode its subfiles from the received messages. For non-leader users, all the required subfiles may not be directly available. However, it can generate these subfiles by taking appropriate linear combinations of the received contents.
 
 \end{proof}

\section{Converse}
 \label{section:converse}
 %\subfile{sections/converse}
 A lower bound for the worst-case delivery time is obtained using index coding techniques. The equivalent index coding problem  $\mathcal{I}(\mathcal{M}_D,\mathcal{L},d)$ (for distinct demands) is obtained as follows: each requested file is split into $2^\Lambda$ disjoint subfiles $W^{d_k}_{\mathcal{S}}$, where $\mathcal{S} \subseteq [\Lambda]$ represents the set of helper nodes in which the parts of the file are cached. There are in total $K2^{\Lambda-1}$ requested subfiles and each bit of it corresponds to a message in the index coding problem. Without loss of generality, assume single unicast case (where each user requests a single, distinct message). 

From the delay obtained in \eqref{e2}, we get an upper bound on the length of the optimal linear index code for the problem $\mathcal{I}(\mathcal{M}_D,\mathcal{L},d)$. That is, we have $\kappa(\mathcal{M}_D, \mathcal{L},d) \leq T_D(\mathcal{L})$, which gives
\begin{equation}
 \begin{aligned}[f]
    \kappa(\mathcal{M}_D, \mathcal{L},d)  \leq & \left(\frac{N-M}{M}\right)\displaystyle\sum_{s=1}^{\Lambda} \left(\displaystyle\sum_{n=1}^{\Lambda-(s-1)}\mathcal{L}_n\binom{\Lambda-n}{s-1}\right).\\
    & \left(\frac{M}{N}\right)^{s}\left(1-\frac{M}{N}\right)^{\Lambda-s}F.
      \label{k}
 \end{aligned}
\end{equation}

Using the relation in \eqref{bound}, we can lower bound $\kappa(\mathcal{M}_D, \mathcal{L},d)$ by the generalized independence number $\alpha(\mathcal{M}_D, \mathcal{L},d)$. A generalized independent set $H$ can be constructed as follows:
\begin{equation*}
    H = \bigcup_{n \in d}\{ W^n_{\mathcal{S}}: 1,2,...,c(n) \notin \mathcal{S} \}
\end{equation*}
\noindent where, $c(n)$ represents the cache to which the user demanding the file $W^n$ is associated to. The claim is that the messages in set $H$ form an acyclic subgraph in the hypergraph representing $\mathcal{I}(\mathcal{M}_D,\mathcal{L},d)$. Consider a set $H' \subseteq H$ such that

\begin{equation*}
    H' = \{W^{n_1}_{\mathcal{S}_1}, W^{n_2}_{\mathcal{S}_2},..., W^{n_k}_{\mathcal{S}_k} \}
\end{equation*}
\noindent where, $c(n_1)\leq c(n_2)\leq...\leq c(n_k)$. Consider the message $W^{n_1}_{\mathcal{S}_1}$ for $|\mathcal{S}_1|= \Lambda-1$ (say), then the user requesting this message does not have any other message in $H'$ as side information, which implies the set of messages in $H'$ form an acyclic subgraph. Thus, any subset of $H$ lies in  $\mathcal{J}(\mathcal{I})$. Hence, we have $\alpha (\mathcal{M}_D,\mathcal{L},d) \geq |H|$. 

The size of the messages in $H$ corresponding to the demands of the users associated to a helper node $n \in [\Lambda]$ is given by 
\begin{equation*}
\displaystyle\sum_{s=0}^{\Lambda-n}\mathcal{L}_n\binom{\Lambda-n}{s}\left(\frac{M}{N}\right)^s\left(1-\frac{M}{N}\right)^{\Lambda-s}F .
\end{equation*}

The limits of the summation came so because $\binom{\Lambda-n}{s}$ is zero for $\Lambda-n < s$. Therefore,
\begin{equation}
 \begin{aligned}[f]
    |H|  = 
      \displaystyle\sum_{n=1}^{\Lambda} \displaystyle\sum_{s=0}^{\Lambda-n}\mathcal{L}_n\binom{\Lambda-n}{s}\left(\frac{M}{N}\right)^{s}\left(1-\frac{M}{N}\right)^{\Lambda-s}F.
 \label{a}
 \end{aligned} 
\end{equation}
Changing the limits of second summation, \eqref{a} becomes
\begin{equation}
 \begin{aligned}[f]
    |H|  = 
      \displaystyle\sum_{n=1}^{\Lambda} \displaystyle\sum_{s=1}^{\Lambda-n+1}\mathcal{L}_n &\binom{\Lambda-n}{s-1}  \left(\frac{M}{N}\right)^{s-1}.\\
      & \left(1-\frac{M}{N}\right)^{\Lambda-(s-1)}F.
 \label{a1}
 \end{aligned} 
\end{equation}
Changing the order of summation in \eqref{a1}

\begin{equation}
 \begin{aligned}[f]
  \textrm{\hspace{-2cm}}   |H|  = & \displaystyle\sum_{s=1}^{\Lambda} \left[ \displaystyle\sum_{n=1}^{\Lambda-s+1}\mathcal{L}_n\binom{\Lambda-n}{s-1}\right].\\
       & \textrm{\hspace{1cm}}\left(\frac{M}{N}\right)^{s-1}\left(1-\frac{M}{N}\right)^{\Lambda-s+1}F.
 \label{amid}
 \end{aligned} 
  \end{equation}

Thus,
\begin{equation}
 \begin{aligned}[f]
    \alpha (\mathcal{M}_D, \mathcal{L},d) \geq  \left(\frac{N-M}{M}\right)&\displaystyle\sum_{s=1}^{\Lambda} \left[\displaystyle\sum_{n=1}^{\Lambda-s+1}\mathcal{L}_n\binom{\Lambda-n}{s-1}\right].\\
     & \left(\frac{M}{N}\right)^{s}\left(1-\frac{M}{N}\right)^{\Lambda-s}F.
 \label{a2}
 \end{aligned} 
\end{equation}

\noindent Hence, from \eqref{bound}, \eqref{k} and \eqref{a2}, we obtain
\begin{equation}
 \begin{aligned}[f]
    \kappa (\mathcal{M}_D, \mathcal{L},d) = & \alpha(\mathcal{M}_D, \mathcal{L},d)\\
       = & \left(\frac{N-M}{M}\right)\displaystyle\sum_{s=1}^{\Lambda} \left[\displaystyle\sum_{n=1}^{\Lambda-s+1}\mathcal{L}_n\binom{\Lambda-n}{s-1}\right].\\
       & \textrm{\hspace{1cm}}\left(\frac{M}{N}\right)^{s}\left(1-\frac{M}{N}\right)^{\Lambda-s}F
 \label{a3}
 \end{aligned} 
\end{equation}
\noindent which is same as the worst-case delivery time obtained by the caching scheme described in Section \ref{section:decen}. This concludes the proof of Theorem 1.

\section{A Comment on \enquote{Online Coded Caching}}
 \label{section:nujoom}
 %\subfile{sections/online_nujoom}
 The coded Least Recently Sent (LRS) caching algorithm for online coded caching \cite{online} considered a dedicated cache network.
The algorithm involves two phases: a delivery phase and a cache update phase. When all the demanded files during the delivery are partially cached, there is no requirement for a cache update. When a file which is not partially cached is demanded, a cache update phase needs to be done. The LRS coded caching algorithm in \cite{online} replaces the new file with the least recently sent file from the users' cache. We observed that the least recently sent file need not be unique in some cases. In those cases, all the users may not be able to replace the same file unless there is cooperation assumed among the users. We suggest a modification to the LRS algorithm by ordering of files so that this problem will be taken care of. This fix will keep the results on the error-correction for the online coded caching scheme \cite{nj} intact. 

\subsection{Example: Non-Uniqueness of the Least Recently Sent File}
\label{subsec:ex}
This example illustrates the fact that the least recently sent file need not be unique. Consider an online coded caching setting with $N=3$ popular files $\{A, B, C\}$, $K = 3$ users, cache memory of $M=1$. Let $\beta = \frac{5}{3}$, so that each user caches a fraction $\frac{1}{5}$ of the bits of $N' = \beta N=5$ files. We assume that initially each user partially caches the files $\{A, B, C, D, E\}$. 
\begin{itemize}
\item At the first time slot, $t=1$, let the demand be $d =(A,B,C).$ Since all the files are partially cached, no cache update is done.
\item At the second time slot, $t=2$, let the demand be $d=(A,B,F).$ Here $F$ is not partially cached. Hence cache update is required. But there are two files $D$ and $E$ which are not sent by the server. Hence, the least recently sent file here is not unique.
\end{itemize}
If we had an ordering of files before placing for the first time, then there would be a unique file to get replaced during the cache update.

\subsection{Ordering of Files}
The ordering of files before getting placed for the first time will make the file to be replaced a unique one. Before discussing this ordering, we state one observation as a lemma below.
\begin{lem}
No ordering is necessary for the files which are not partially cached before the first time slot.
\end{lem}
\begin{IEEEproof}
The files which are not partially cached are sent uncoded during the delivery. If these are demanded at different time slots, they will be updated at different time slots and hence there is no ambiguity among these files regarding the order. If these files are demanded at the same time slot, then the server will send these files uncoded one after the other. Hence the order can be taken as the order in which the server sends these files.
\end{IEEEproof}
Thus ordering is required for the files which are cached at the users before the first time slot. The ordering can be formally discussed as follows: Let $X_1, X_2, \ldots, X_{N^{'}}$ be the files that are to be partially cached at the users. Assign a unique ordering parameter $o \in \{1,2, \ldots, N'\}$ to each file. This can be chosen based on their popularity or can be chosen randomly. Once this is done, the LRS scheme in \cite{online} can be followed. Whenever there is a situation where there are two or more files that can be replaced with a newly arrived file, this ordering parameter is compared and the file with the least $o$ can be selected for replacement. This will be the same for all the users and the uniqueness is now guaranteed. This can be illustrated through the following two examples.\\

\noindent \textit{Example 1:} Consider the example given in Section \ref{subsec:ex}. Initially, each user partially caches the files $\{A, B, C, D, E\}$. Assign a unique ordering parameter to each of these files. Assume that the parameters assigned are: $1$ to $E$, $2$ to $D$, $3$ to $C$, $4$ to $B$ and $5$ to $A$. During the second time slot, out of $D$ and $E$, now the users can choose $E$ to be replaced with the newly arrived file $F$ without any ambiguity. \\

\noindent \textit{Example 2:} Consider an online coded caching problem with $N=3$ popular files ${A,B}$, $K=2$ users, cache memory of $M=1$. Let $\beta=3$, so that each user caches a fraction $\frac{1}{6}$ of the bits of $N^{'}= \beta N= 6$ files. Let the files which are partially cached initially are $\{A, B, C, D, E, F \}.$ Let the ordering parameters assigned be $1$ to $A$, $2$ to $B$, $3$ to $C$, $4$ to $D$, $5$ to $E$ and $6$ to $F$. 
\begin{itemize}
\item At the time slot $t=1$, let the demand be $d=(A,B).$ No cache update is required.
\item At the time slot $t=2$, let the demand be $d=(C,D).$ No cache update is required.
\item At the time slot, $t=3$, let the demand be $d=(E,F).$ No cache update is required.
\item At the time slot, $t=4$, let the demand be $d=(E,G).$ The file $G$ is a new arrival and hence a cache update is required here. The least recently sent files are $A$ and $B$. Since there are two of them, the ordering parameter has to be taken into consideration. The one with the least ordering parameter here is $A$ and hence $A$ is replaced with $G$ by all the users without any ambiguity. 
\end{itemize}

\section{Online Coded Caching with Shared Caches}
  \label{section:onl}
  %\subfile{sections/online_part}
  In this section, we propose an online coded caching algorithm for the shared cache network. The setting is as follows: In each time slot $t$, the server maintains a set of popular files $\mathcal{N}_t$, where $|\mathcal{N}_t|=N$. The set of helper nodes $\Lambda$ (of size $M$) caches from a list of $N'=\beta N$ files in a decentralized manner, for some $\beta \geq 1$. The placement phase is same as that in Section \ref{section:decen}.A, but we assume an ordering for the initial set of cached files. The users get connected to one of the caches, and let the association be $\mathcal{U}$ with profile $\mathcal{L}$. The set $\mathcal{N}_t$ evolves over time and each user $k \in [K]$ demands a file $W^{d_t(k)}$ from the set $\mathcal{N}_t$, uniformly at random. Assume distinct demand case, and the demand vector at time $t$ is represented as $d_t=(d_t(1), d_t(2),..., d_t(K))$. 

In the case of online algorithms, the delivery procedure has to account for an additional scenario where some or all of the users' requested files are not partially cached. That means, the cache contents are not updated with the current set of popular files. In such scenarios, all those requested files that are not partially cached are sent uncoded over the shared link and the rest of the requests are satisfied using the delivery scheme described in Algorithm 1 (which is considered as the equivalent offline scheme).

After the delivery phase, the caches are updated as follows: For any user $k$, if the requested file $d_t(k)$ is not partially cached, then all the helper nodes evict the least-recently sent file contents and instead, caches $MF/N'$ bits of the file $d_t(k)$. This is possible since the file $d_t(k)$ is sent uncoded over the link during the transmissions. The eviction policy guarantees that the same set of files are cached at all the nodes. The coded LRS algorithm for the shared cache network is summarized in Algorithm 2.

\begin{algorithm}
\renewcommand{\thealgorithm}{2}
\caption{Coded LRS Caching algorithm for Shared Caches}
\begin{algorithmic}[1]
\Procedure{Delivery}{}  
     \For{\texttt{$k \in [K]$}}
       \If{$W^{d_t(k)}$ is not partially cached}
 
         \State Send the entire file $W^{d_t(k)}$, uncoded and
         \State Modify $\mathcal{U} \rightarrow \mathcal{U}^{'}$ and $\mathcal{L} \rightarrow \mathcal{L}^{'}$, after removing the user $k$.
       \EndIf
    \EndFor
    \For{\texttt{$j = 1:\mathcal{L}^{'}_1$}}
     \For{\texttt{$s = \Lambda, \Lambda-1,...,1$}}
       \For{\texttt{$\mathcal{S} \subseteq [\Lambda]: |\mathcal{S}|=s$}}
       
          \State  Transmit $\underset{{\lambda \in \mathcal{S}: \mathcal{L}^{'}_\lambda \geq j}}{\bigoplus}W_{\mathcal{S}\backslash \{\lambda\}}^{d_t({\mathcal{U}^{'}_{\lambda}(j))}}$
      \EndFor
     \EndFor
    \EndFor
\EndProcedure
\Procedure{Cache Update}{}
  \For{\texttt{$k \in [K]$}}
  \If{$W^{d_t(k)}$ is not partially cached}
 
    \State The least recently requested file contents are replaced from all the caches with a random subset of $MF/N'$ bits of file $W^{d_t(k)}$.
  \EndIf
  \EndFor
\EndProcedure
     
\end{algorithmic}
\end{algorithm}

We assume that the set $\mathcal{N}_t$ from which the demand vector $d_t$ is derived, remains unchanged over the course of delivery. Decoding is similar to the offline scheme.

\textit{Performance Analysis}: The number of time slots/ the delivery time required to satisfy the demand $d_t$ for an association profile $\mathcal{L}$ is represented as $T_{online}(\mathcal{L})$, where all the demands in $d_t$ are distinct. The optimal delivery time is denoted as $T_{online}^{*}(\mathcal{L})$.
There are two situations arising in online algorithms:
\begin{itemize}
    \item \textit{Case I}: When all the requested files are partially cached, the normalized delivery time required to satisfy $d_t$ is
    \begin{equation}
        T_{online}^{*}(\mathcal{L}) = T_D(\mathcal{L}).
    \end{equation}
    
    \item \textit{Case II}: In $d_t$, if $U$ number of requested files are not cached in any of the nodes, the normalized delivery time is
     \begin{equation}
        T_{online}^{*}(\mathcal{L}) = T_D(\mathcal{L}^{'}) + U
        \label{delay_online}
    \end{equation}
\noindent where, $\mathcal{L}^{'}$ is the modified user-to-cache association obtained after removing $U$ users from $\mathcal{L}$.
    
\end{itemize}

\noindent \textit{Converse}: The LRS online caching scheme for shared caches, given in Algorithm $2$ is optimal.

Since Case I is same as the offline scheme, it is enough to show the optimality for Case II. The index coding problem associated with this case is represented as $\mathcal{I}(\mathcal{M}_{LRS},\mathcal{L},d_t)$, which consists of 
$((K-U)2^{\Lambda-1} + U)F$ messages wanted by $K$ users. Assume unicast case. The min-rank and independence number are denoted as $\kappa(\mathcal{M}_{LRS}, \mathcal{L},d_t)$ and $\alpha(\mathcal{M}_{LRS}, \mathcal{L},d_t)$, respectively.

From \eqref{delay_online}, we get
\begin{equation*}
    \kappa(\mathcal{M}_{LRS}, \mathcal{L},d_t) \leq  T_D(\mathcal{L}^{'})F + UF 
\end{equation*}
\noindent which gives 
\begin{equation}
\begin{aligned}[f]
    \kappa(\mathcal{M}_{LRS}, \mathcal{L},d_t) \leq & \left(\frac{N'-M}{M}\right)\displaystyle\sum_{s=1}^{\Lambda} \left[\displaystyle\sum_{n=1}^{\Lambda-s+1}\mathcal{L}_n^{'}\binom{\Lambda-n}{s-1}\right].\\
     & \left(\frac{M}{N'}\right)^{s}\left(1-\frac{M}{N'}\right)^{\Lambda-s}F + UF.
     \label{onkappa}
     \end{aligned}
\end{equation}

Construct a set $H$ as follows:
\begin{equation*}
    H =  \left(\bigcup_{n \in d\backslash d_n}\{ W^n_{\mathcal{S}}: 1,2,...,c(n) \notin \mathcal{S} \}\right) \bigcup_{n \in d_n}{W^n}
\end{equation*}
\noindent where, $d_n$ represents the indices of the requested files that are not partially cached. Let $A= \bigcup_{n \in d\backslash d_n}\{ W^n_{\mathcal{S}}: 1,2,...,c(n) \notin \mathcal{S} \}$  and $B = \bigcup_{n \in d_n}{W^n}$. Messages in set $B$ are not available as side-information to any receivers. Observe that the set $A$ is a generalized independent set (it corresponds to the offline scheme with $K-U$ users having a profile $\mathcal{L}^{'}$ ). Hence, it follows directly that $H$ is a generalized independent set and $\alpha(\mathcal{M}_{LRS}, \mathcal{L},d_t) \geq |H|$. Therefore,
\begin{equation}
\begin{aligned}[f]
    \alpha(\mathcal{M}_{LRS}, \mathcal{L},d_t) \geq & \left(\frac{N'-M}{M}\right)\displaystyle\sum_{s=1}^{\Lambda} \left[\displaystyle\sum_{n=1}^{\Lambda-s+1}\mathcal{L}_n^{'}\binom{\Lambda-n}{s-1}\right].\\
     & \left(\frac{M}{N'}\right)^{s}\left(1-\frac{M}{N'}\right)^{\Lambda-s}F + UF.
     \label{onlalpha}
     \end{aligned}
\end{equation}

Thus from \eqref{bound}, \eqref{onkappa} and \eqref{onlalpha}, the optimality follows. This concludes the proof of Theorem 2.

\section{Optimal Linear Error Correcting Delivery Scheme for Decentralized Coded Caching with Shared Caches}
   \label{section:errorc}
   %\subfile{sections/errorcorr}
   Consider the case where the shared link between the server and the users is error-prone. Each user should be able to decode their desired files even in the presence of a finite number of transmission errors. In the converse proofs for both offline and online schemes, it is shown that
\begin{equation}
 \alpha(\mathcal{M},\mathcal{L},d) = \kappa(\mathcal{M},\mathcal{L},d)
 \label{ec}
 \end{equation}
\noindent where, $\mathcal{M}$ represents the placement schemes $\mathcal{M}_D$ and $\mathcal{M}_{LRS}$. (In online case, $d$ corresponds to the demand vector at a particular time instant). Therefore by \eqref{ebound}, we can obtain an optimal linear error correcting delivery scheme by concatenating the corresponding delivery scheme with an optimal linear error correcting code of required minimum distance. Let $T_D^{*}(\mathcal{L},\delta)$ denotes the optimal worst-case delivery time taken by a $\delta-$error correcting caching scheme (in offline case). Therefore, 
 
\begin{equation*}
       T_D^{*}(\mathcal{L},\delta) = \frac{N_q[\kappa(\mathcal{M}_D,\mathcal{L},d),2\delta+1]}{F}.
 \end{equation*}
   
Decoding is done using syndrome decoding proposed for error correcting index codes \cite{r5}. The same applies to online case also.

\section{Examples}
  \label{section:ex}
  %\subfile{sections/examples}
  In this section, we present examples to illustrate the caching schemes described in Sections \ref{section:decen} and \ref{section:onl}.\\

\noindent \textit{Example 1: $K = N = 4$, $\Lambda = 2$ and $\mathcal{L}=(3,1)$}.

Consider a case with $K = 4$ users connected to $\Lambda=2$ helper caches, each of size $M = 2$ units of file, storing contents from a set of $N=4$ files $\{W^1, W^2, W^3, W^4\}$. Each file is of size $F$ bits.

In the placement phase, each helper node independently caches $\frac{F}{2}$ bits of each file, chosen uniformly at random. The placement procedure partitions the file $W^n$ into $4$ disjoint subfiles $\{W^n_\phi, W^n_{1}, W^n_{2}, W^n_{\{1,2\}}\}$. For sufficiently large $F$,
\begin{equation*}
 \begin{aligned}
    |W^n_\phi| = |W^n_{1}| =  |W^n_{2}| = |W^n_{\{1,2\}}| \approx \frac{F}{4}
    \end{aligned}
\end{equation*}

The placement phase is followed by user-to-cache assignment phase where 
users $\mathcal{U}_1=\{1,2,3\}$ and $\mathcal{U}_2=\{4\}$ are assigned to caches $1$ and $2$, respectively. This association leads to a $\mathcal{L}=(3,1)$. Let the demand vector be $d=(1,2,3,4)$.

Delivery consists of $\mathcal{L}_1=3$ rounds, where each round serves the following set of users\\
\begin{equation*}
    \begin{aligned}
     \mathcal{R}_1 &= \{1,4\}\\
     \mathcal{R}_2 &= \{2\}\\
     \mathcal{R}_3 &= \{3\}\\
    \end{aligned}
\end{equation*}
 
 The transmissions in three rounds are as follows:
 \begin{equation*}
    \begin{aligned}
      \textrm{First round} &:     W^1_2 \oplus W^4_1, W^1_\phi, W^4_\phi\\
       \textrm{Second round} &:  W^2_2, W^2_\phi\\
       \textrm{Third round} &:   W^3_2, W^3_\phi
   \end{aligned}
\end{equation*}

Thus, the normalized delivery time required to serve all the users is $T_D(3,1) = \frac{7}{4}$.

\textit{Optimality}: The index coding problem $\mathcal{I}(\mathcal{M}_D,\mathcal{L},d)$ consists of $8F$ messages and $8F$  receivers. For the side-information graph representing $\mathcal{I}(\mathcal{M}_D,\mathcal{L},d)$, we have $\kappa(\mathcal{M}_D,\mathcal{L},d) \leq \frac{7}{4}F$. Construct a generalized independent set $H$ as follows

\begin{equation*}
    H = \{W^1_2, W^1_\phi, W^2_2, W^2_\phi, W^3_2, W^3_\phi, W^4_\phi\}.
\end{equation*}

Thus, we obtain $\alpha(\mathcal{M}_D,\mathcal{L},d)\geq |H| =\frac{7}{4}F$ which gives $\kappa(\mathcal{M}_D,\mathcal{L},d) = \alpha(\mathcal{M}_D,\mathcal{L},d) = \frac{7}{4}F$.

\textit{Error Correction}: If we need to correct $\delta = 1$ error, the transmissions need to be concatenated with an optimal classical error-correcting code of length $N_2[7,3]$ (by taking $F =4$). From \cite{grassl}, we get $N_2[7,3] = 11$. Therefore, the optimal delivery time in the case of single error correction is $11/4$  (normalized with respect to file size). Syndrome decoding is used to recover the messages.

\textit{Comparison}: For this example, the performance of centralized schemes are as follows:
\begin{center}
$T(3,1) = \frac{3}{2}$ \hspace{1cm}   from \cite{r2} (uncoded placement).\\
\end{center}
\begin{center}
$T(3,1) = \frac{3}{2}$ \hspace{1.4cm}   from \cite{r8} (coded placement).\\
\end{center}

The delivery time achieved by the schemes in \cite{r2} and \cite{r8} happen to be the same for this memory, and is less than $T_D(3,1)$. This is a small penalty incurred due to decentralization.\\

\noindent \textit{Considering the case of non-distinct demands}\\

Consider a demand vector $d=(1,2,2,1)$. Out of the $4$ files, only two of them are requested by the users. First, we will remove the redundancy present within a cache.  User $3$ is not considered in the delivery process, since it can acquire all the desired subfiles from the transmissions meant for user $1$. Hence, the modified profile $\mathcal{L}^{'}=(2,1)$.

\textit{Transmissions:} Initially transmit $W^1_\phi$ and $W^2_\phi$. The remaining subfiles are transmitted in two rounds, where each round serves the following users
\begin{equation*}
    \begin{aligned}
     \mathcal{R}_1^{'}=\{1,4\} \textrm{\hspace{0.15cm}and\hspace{0.15cm}}
     \mathcal{R}_2^{'}=\{2\}
    \end{aligned}
\end{equation*}

The set of leader users in each round are: $\mathcal{Q}_1 = \{1\}$ and $\mathcal{Q}_2 = \{2\}$. The remaining transmissions are:
\begin{equation*}
    \begin{aligned}
       \textrm{First round} &:   W^1_2 \oplus W^1_1\\ 
      \textrm{Second round} &:   W^2_2\\ 
     \end{aligned}
\end{equation*}

\noindent The normalized delivery time required in this case is $T_D(3,1)=1$.\\

\noindent \textit{Effect of user-to-cache association on performance}\\

For this example, the possible association profiles are: $\mathcal{L}=(4,0)$, $\mathcal{L}=(3,1)$ and $\mathcal{L}=(2,2)$. 
Consider the user-to-cache associations $\mathcal{U}=\{\{1,2,3,4\},\phi\}$ and $\mathcal{U}=\{\{1,2\},\{3,4\}\}$ having profiles $\mathcal{L}=(4,0)$ and $\mathcal{L}=(2,2)$, respectively.
Let $d=(1,2,3,4)$. Then the transmitted messages in both the cases are as follows:\\
For $\mathcal{U}=\{\{1,2,3,4\},\phi\}$
\begin{equation*}
    \begin{aligned}
    \textrm{First round} &:  W^1_2, W^1_\phi \\ 
    \textrm{Second round} &: W^2_2, W^2_\phi \\ 
    \textrm{Third round} &: W^3_2, W^3_\phi \\ 
    \textrm{Fourth round} &: W^4_2, W^4_\phi \\ 
    \end{aligned}
\end{equation*}

For $\mathcal{U}=\{\{1,2\},\{3,4\}\}$
\begin{equation*}
    \begin{aligned}
       \textrm{First round} &:   W^1_2 \oplus W^3_1, W^1_\phi, W^3_\phi \\ 
      \textrm{Second round} &:   W^2_2 \oplus W^4_1, W^2_\phi, W^4_\phi \\ 
     \end{aligned}
\end{equation*}

The normalized delivery time achieved in each case is listed below. The performance of the centralized caching scheme is also provided for a comparison.

\begin{table}[h]

\centering
\begin{tabular}{ c c c c } % creating eight columns
\hline %inserting double-line
  & $\mathcal{L}=(4,0)$ & $\mathcal{L}=(3,1)$ & $\mathcal{L}=(2,2)$  \\ [0.5ex]
 % \cline{2-5}
 \hline
$T_D(\mathcal{L})$ & 2 & 7/4 & 3/2 \\ % Entering row contents
 
$T(\mathcal{L})$ \cite{r2} & 2 & 3/2 & 1\\[0.5ex] % [1ex] adds vertical space
\hline % inserts single-line
\end{tabular}

\end{table}

\begin{itemize}
\item When all the users are connected to the same cache, there is no coded multicasting gain, only local caching gain $K\left(1-\frac{M}{N}\right)$ is achieved. Both the centralized and decentralized schemes exhibit the same performance for this profile.

\item Uniform assignment yields the smallest delivery time amongst all profiles.\\

\end{itemize}

\begin{figure}[t!]
\begin{center}
 \includegraphics[width=\columnwidth]{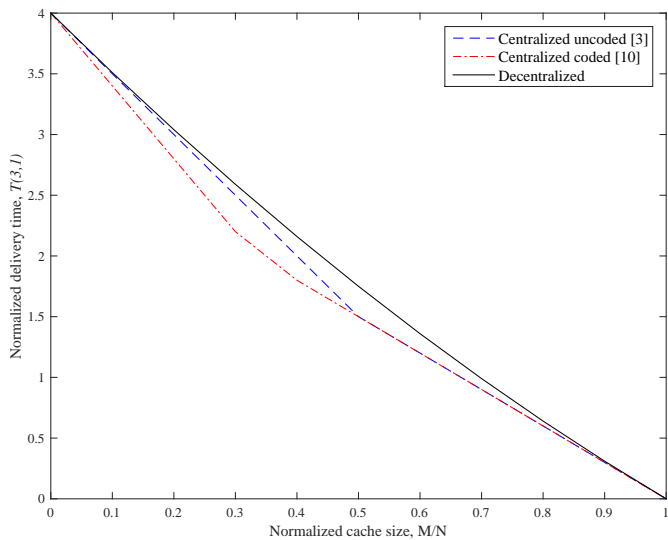}
 \caption{Performance of different caching schemes for a system with $K=4$ users, $N=4$ files $\Lambda=2$ caches and $\mathcal{L}=(3,1)$, for worst-case demand.}
 \end{center}
\end{figure}

\noindent \textit{Example 2: Online Setting - $K = 4$, $N = 4$ popular files, $\Lambda = 2$ and $\mathcal{L}=(3,1)$}.

Consider a system with $N=4$ popular files, $K=4$ users connected to $\Lambda=2$ caches, each of size $M=2$ units of file. Let $\beta=5/4$, then the helper nodes caches from a list of $N'=\beta N=5$ files, each of size $F$ bits. Assume that initially the helper nodes caches $2F/5$ bits from the files $\{W^1, W^2, W^3, W^4, W^5\}$, chosen uniformly at random. Let the order assigned be $(1,2,3,4,5)$. Users $\mathcal{U}_1=\{1,2,3\}$ and $\mathcal{U}_2=\{4\}$ gets connected to caches $1$ and $2$, respectively.

At time $t_1$: Let the set of popular files be $\mathcal{N}_{t_1}=\{ W^2, W^3, W^4, W^5\}$. Users' demand from the set $\mathcal{N}_{t_1}$, and assume distinct demand case. Let $d_{t_1}=(2,3,4,5)$. Since all the requested files are partially cached, the situation is similar to the offline scheme. We get the normalized delivery time as $T_{online}(3,1)=54/25$. Since all the requested files are partially cached, there is no change in the cache contents.

At time $t_2$: A new arrival ($W^6$) happens, so the set of popular files needs to be modified. The file $W^5$ is randomly chosen and replaced with the new file $W^6$, so that $\mathcal{N}_{t_2} = \{ W^2, W^3, W^4, W^6\}$. Take the demand vector to be $d_{t_2}=(6,2,3,4)$. Except $W^6$, all the other files are partially stored in all caches. So, the delivery is as follows:
\begin{itemize}
    \item Transmit the entire file $W^6$ as it is.
    \item Then, $\mathcal{U}^{'}_1=\{2,3\}$ and $\mathcal{U}^{'}_2=\{4\}$. Therefore, $\mathcal{L}^{'}=(2,1)$. The remaining transmissions are:
    \begin{equation*}
     \begin{aligned}
    \textrm{First round} &:    W^2_2 \oplus W^4_1, W^2_\phi, W^4_\phi \\ 
    \textrm{Second round} &: W^3_2, W^3_\phi 
     \end{aligned}
    \end{equation*}
    
\end{itemize}
Thus, in this case $T_{online}(3,1)=39/25+1=64/25$.\\
\textit{Cache Update:} Since $W^1$ is the least recently requested file amongst the set of cached files, it is evicted from all the helper nodes and replaced with the bits from file $W^6$. Here, we do need to consider the order because there is only one least recently sent file. Hence the set of cached files becomes $\{W^2, W^3, W^4, W^5, W^6\}$.

\textit{Converse}: Here, we show that the delivery time achieved in time $t_2$ is optimal. Consider the index coding problem $\mathcal{I}(M_{LRS},\mathcal{L},d_{t_2})$. We have $\kappa(\mathcal{M}_{LRS},\mathcal{L},d_{t_2}) \leq 64F/25$.
Construct a generalized independent set $ H=\{W^2_2, W^2_\phi, W^3_2,W^3_\phi, W^4_\phi, W^6\}$, which gives 
 $\alpha(\mathcal{M}_{LRS},\mathcal{L},d_{t_2}) \geq 64F/25$.
 Thus, $\kappa(\mathcal{M}_{LRS},\mathcal{L},d_{t_2})=\alpha(\mathcal{M}_{LRS},\mathcal{L},d_{t_2})$ and the optimality follows.

\section{Conclusion}
  \label{section:end}
  %\subfile{sections/conclusion}
  A decentralized coded caching scheme is introduced for the shared cache network where each cache serves an arbitrary number of users. It is extended further to the online setting. An index-coding based converse is presented to prove that the performance of the scheme is optimal for distinct demands, under uncoded placement in both cases. The converse  enabled the  construction of optimal linear error-correcting delivery schemes for the above scenarios.

\section*{Appendix}
 \label{section:app}
 %\subfile{sections/appendix}
\subsection{Derivation of the delivery time for a uniform profile}
For an arbitrary association profile,
\begin{equation}
  \begin{aligned}[f]
      T_D^{*}(\mathcal{L}) & = \left(\frac{N-M}{M}\right)\times\\
     & \displaystyle\sum_{s=1}^{\Lambda}\left[\displaystyle\sum_{n=1}^{\Lambda-(s-1)} \mathcal{L}_n\binom{\Lambda-n}{s-1}\right]
     \left(\frac{M}{N}\right)^s\left(1-\frac{M}{N}\right)^{\Lambda-s}. 
     \label{arb}
\end{aligned}
 \end{equation}
 If the association is uniform, then $\mathcal{L}_{unif.}=\left(\frac{K}{\Lambda}, \frac{K}{\Lambda},..., \frac{K}{\Lambda}\right)$. The optimal delivery time for this case is
 \begin{equation}
  \begin{aligned}[f]
      T_D^{*}(\mathcal{L}_{unif.}) & = \left(\frac{N-M}{M}\right)\left(\frac{K}{\Lambda}\right)\times\\
     & \displaystyle\sum_{s=1}^{\Lambda}\left[\displaystyle\sum_{n=1}^{\Lambda-(s-1)}\binom{\Lambda-n}{s-1}\right]
     \left(\frac{M}{N}\right)^s\left(1-\frac{M}{N}\right)^{\Lambda-s}. \label{app}
\end{aligned}
 \end{equation}

\noindent By hockey stick identity:
$
    \displaystyle\sum_{n=1}^{\Lambda-(s-1)}\binom{\Lambda-n}{s-1} = \binom{\Lambda}{s}
$, then \eqref{app} becomes 
\begin{equation}
  \begin{aligned}[f]
      T_D^{*}(\mathcal{L}_{unif.}) & = \left(\frac{N-M}{M}\right)
     \left(\frac{K}{\Lambda}\right).\\
     & \textrm{\hspace{2cm}}\displaystyle\sum_{s=1}^{\Lambda}\binom{\Lambda}{s}
     \left(\frac{M}{N}\right)^s\left(1-\frac{M}{N}\right)^{\Lambda-s} \\
     & = \left(\frac{N-M}{M}\right)
     \left(\frac{K}{\Lambda}\right).\\
     & \textrm{\hspace{-1cm}}\left[\displaystyle\sum_{s=0}^{\Lambda}\binom{\Lambda}{s}
     \left(\frac{M}{N}\right)^s \left(1-\frac{M}{N}\right)^{\Lambda-s}-\left(1-\frac{M}{N}\right)^{\Lambda}\right].\\ 
\end{aligned}
 \end{equation}
 
 Thus, the normalized delivery time required for uniform user-to-cache association is
 \begin{equation}
  \begin{aligned}[f]
      T_D^{*}(\mathcal{L}_{unif.}) & = 
      \left(\frac{N-M}{M}\right)
     \left(\frac{K}{\Lambda}\right)\left[1- \left(1-\frac{M}{N}\right)^\Lambda\right] .
\end{aligned}
 \end{equation}
 
Same delivery time was obtained using the alternate algorithm in [2, Section V.B]. Thus, we prove that their scheme is optimal. 

\subsection{Delivery time when all users are assigned to a single cache}
In this case, we have  $\mathcal{L}=(K,0,..0)$. Then \eqref{arb} becomes
\begin{equation}
  \begin{aligned}[f]
       T_D^{*}(\mathcal{L}) & = \left(\frac{N-M}{M}\right)
      \displaystyle\sum_{s=1}^{\Lambda}K\binom{\Lambda-1}{s-1}
     \left(\frac{M}{N}\right)^s\left(1-\frac{M}{N}\right)^{\Lambda-s} \\
     & = K\displaystyle\sum_{s=1}^{\Lambda}\binom{\Lambda-1}{s-1}
     \left(\frac{M}{N}\right)^{(s-1)}\left(1-\frac{M}{N}\right)^{\Lambda-(s-1)} 
 \end{aligned}
\end{equation}
Which can be written as
\begin{equation}
  \begin{aligned}[f]
       T_D^{*}(\mathcal{L}) &= K\left(1-\frac{M}{N}\right)\displaystyle\sum_{s=0}^{\Lambda-1}\binom{\Lambda-1}{s}
     \left(\frac{M}{N}\right)^{s}\left(1-\frac{M}{N}\right)^{(\Lambda-1)-s}\\
     & = K\left(1-\frac{M}{N}\right).
 \end{aligned}
\end{equation}

\section*{Acknowledgement}
This work was supported partly by the Science and Engineering Research Board (SERB) of Department of Science and Technology (DST), Government of India, through J.C Bose National Fellowship to B. Sundar Rajan.


\begin{thebibliography}{1}
	
	\bibitem{r1} 
	M. A. Maddah-Ali and U. Niesen, \enquote{Fundamental limits of Caching,} \textit{IEEE Trans. Inf. Theory}, vol. 60, no. 5, pp. 2856–2867, May 2014.
	
\bibitem{r3} 
	M. A. Maddah-Ali and U. Niesen, \enquote{Decentalized Coded Caching Attains Order-Optimal Memory-Rate Tradeoff,} \textit{IEEE/ACM Trans. Netw.,} vol. 23, no. 4, pp. 1029-1040, Aug. 2014.
	
	\bibitem{r2}
	E. Parinello, A. Unsal, and P. Elia, \enquote{Coded Caching with Shared Caches: Fundamental limits with Uncoded Prefetching,} 2018, \textit{arXiv: 1809.09422.}[Online]. Available: https://arxiv.org/abs/1809.09422.
	
	\bibitem{online}
	R. Pedarsani, M. A. Maddah-Ali and U. Niesen, \enquote{Online Coded Caching,} \textit{IEEE/ACM Trans. Netw.,} vol. 24, no. 2, pp. 836-845, Apr. 2016.
	

	\bibitem{r6}
	N. S. Karat, A. Thomas, and B. S. Rajan, \enquote{Optimal Error Correcting Delivery Scheme for Coded Caching with Symmetric Batch Prefetching,} in \textit{Proc. 2018 IEEE International Symposium on Information Theory (ISIT)}, Vail, CO, 2018, pp. 2092-2096.
	
	\bibitem{errorc}
	N. S. Karat, A. Thomas and B. S. Rajan, \enquote{Optimal Error Correcting Delivery Scheme for Optimal Coded Caching Scheme with Small Buffers,} in \textit{Proc. 2018 IEEE International Symposium on Information Theory (ISIT)}, Vail, CO, 2018, pp. 1710-1714. 
	
	\bibitem{r7}
	N. S. Karat, S. Dey, A. Thomas, and B. S. Rajan, \enquote{An Optimal Linear Error Correcting Delivery Scheme for Coded Caching with Shared Caches,} in \textit{Proc. IEEE Int. Symp. Inf. Theory (ISIT)}, Jul. 2019, pp. 1217-1221.
	
	 \bibitem{an}
	S. Rathi, A. Thomas, and M. Dutta, \enquote{An Optimal Linear Error Correcting Scheme for Shared Caching with Small Cache Sizes,} in \textit{Proc. IEEE Int. Symp. Inf. Theory (ISIT)}, Jun. 2020, pp. 1676-1680.
	
	\bibitem{nj}
	N. S. Karat, K. L. V. Bhargav, and B. S. Rajan, \enquote{On the Optimality of Two Decentralized Coded Caching Schemes With and Without Error Correction,} in \textit{Proc. IEEE Int. Symp. Inf. Theory (ISIT)}, Jun. 2020, pp. 1670-1675.
	
	\bibitem{r8}
	A. M. Ibrahim, A. A. Zewail, and A. Yener, \enquote{Coded Placements for Systems with Shared Caches,} \textit{ICC 2019 - 2019 IEEE International Conference on Communications (ICC)}, Shanghai, China, 2019, pp. 1-6. 
	
	\bibitem{r4}
	E. Parinello, A. Unsal, and P. Elia, \enquote{Fundamental Limits of Coded Caching with Multiple Antennas, Shared Caches and Uncoded Prefetching,} \textit{IEEE Trans. Inf. Theory}, vol. 66, no. 4, pp. 2252–2268, April 2020.
	

	
	\bibitem{r5}
	S. H Dau, V. Skachek, and Y. M. Chee, \enquote{Error Correction for Index Coding with Side Information,} \textit{IEEE Trans. Inf. Theory}, vol. 59, no. 3, pp. 1517-1531, Mar. 2013.
	
	\bibitem{iscod}
	Y. Birk and T. Kol, \enquote{Coding-on-demand by an informed souce(ISCOD) for efficient broadcast of different supplemental data to caching clients,} \textit{IEEE Trans. Inf. Theory,} vol. 52, no. 6, pp. 2825-2830, Jun. 2006.
	
	\bibitem{index}
	N. Alon, A. Hassidim, E. Lubetzky, U.Stav and A. Weinstein, \enquote{Broadcasting with Side Information,} in \textit{Proc. 49th Annu. IEEE Symp. Found. Comp. Sci.}, Oct. 2008, pp. 823-832.
	
	
	
	\bibitem{grassl}
	M. Grassl, \enquote{Bounds on the minimum distance of linear codes and quantum codes,} Online available at http://www.codetables.de, 2007.
	

  \end{thebibliography}
\end{document}